\documentclass[a4paper,12pt]{article}

\usepackage{tikz}
\usetikzlibrary {calc}
\usetikzlibrary {arrows.meta}
\usetikzlibrary {shapes.geometric}
\usetikzlibrary {graphs,shapes.geometric}
\usetikzlibrary {decorations.pathmorphing}

\usetikzlibrary {graphs}

\usepackage[shortlabels]{enumitem}
\usepackage[]{geometry}
\usepackage[utf8]{inputenc}
\usepackage{fancybox, calc}
\usepackage{authblk}
\usepackage{amssymb}
\usepackage{amsmath}
\usepackage{hyperref}
 \usepackage{systeme}
\usepackage{graphicx}
\usepackage{enumitem}
\usepackage{ucs}
\usepackage{amsfonts}
\usepackage{amssymb}
\usepackage{makeidx}
\usepackage{diagbox}
\usepackage{tabularx}
\usepackage{float}
\usepackage{url}
\usepackage{ marvosym }
\usepackage{ wasysym }
\usepackage{ mathrsfs }
\usepackage[utf8]{inputenc}
\usepackage[english]{babel}
\usepackage[normalem]{ulem}
\usepackage{ucs}
\usepackage{amsfonts}
\usepackage{dcolumn}
\usepackage{makeidx}
\usepackage{bm}
\usepackage{multirow}
\usepackage{graphicx}
\usepackage{hyperref}
\usepackage{tikz}
\usepackage{subfig}
\usepackage{booktabs}
\usepackage{xcolor}
\usepackage{fancyhdr}
\usepackage{amsmath, amsthm, amssymb}
\usepackage{multicol}
\usepackage{yfonts}
\usepackage{multicol}
\usepackage{caption}
\usepackage{cancel}
\usepackage{amssymb,amsfonts,amssymb,bbm}
\usepackage{amsmath}
\usepackage{phaistos,hieroglf,staves,ifsym,marvosym,skull}
\usepackage{tikz}
\usepackage{transparent}
\usepackage{eso-pic}
\usepackage{soul}
\usepackage{cite}
\usepackage{lettrine}

\definecolor{rougeG}{rgb}{.76,0,.12}
\definecolor{vertG}{rgb}{.07,.56,.25}

\numberwithin{equation}{section}
\def\Rset{\mathbb{R}}
\def\pdf{f}
\def\supp{\Omega}
\newcommand{\down}[1]{\mathfrak {D}_{#1}}

\def\Rset{\mathbb{R}}

\def\pdf{f}

\newcommand{\descort}[2][]{\mathfrak{E}_{#1}\ifthenelse{\isempty{#2}}{}{ {\left[ #2 \right]}}}
\newcommand{\escort}[2][]{\varepsilon_{#1}\ifthenelse{\isempty{#2}}{}{ {\left[ #2 \right]}}}

\newtheorem{theorem}{Theorem}[section]

\newtheorem{lemma}{Lemma}[section]
\newtheorem{definition}{Definition}[section]
\newtheorem{proposition}{Proposition}[section]

\numberwithin{equation}{section}
\restylefloat{table}

\title{Entropies, cross-entropies and Rényi divergence: sharp three-term inequalities for probability density functions}
\author[1]{Razvan Gabriel Iagar}
\affil[1]{Departamento de Matemática Aplicada, Ciencia e Ingeniería de los Materiales y Tecnología Electrónica, Universidad Rey Juan Carlos,
		28933 Móstoles (Madrid), Spain}
\author[1,2]{David Puertas-Centeno}

\affil[2]{Data, Complex Networks and Cybersecurity Research Institute, Universidad Rey Juan Carlos, 28028 (Madrid), Spain}

\date{\today}

\begin{document}

\maketitle

\begin{abstract}
A new sharp inequality featuring the differential R\'enyi entropy, the R\'enyi divergence and the R\'enyi cross-entropy of a pair of probability density functions is established. The equality is reached when one of the probability density function is an escort density of the other. This inequality is applied, together with a general framework of a pair of transformations reciprocal to each other, to derive a number of further inequalities involving both classical and new informational functionals. A remarkable fact is that, in all these inequalities, the R\'enyi divergence of two probability density functions is sharply bounded by quotients of informational functionals of cross-type and single type. More precisely, we derive sharp inequalities composed by relative and cross versions of the absolute moments, or of the Fisher information measures (among others), and involving two and three probability density functions.
\end{abstract}


\section{Introduction}

The notions of entropy, divergence and cross-entropy are basic constituent elements in information theory. While the entropy is applied to a single probability density, both the divergence and the cross entropy depend on a couple of probability density functions. A simple and well-known addition relation connects the corresponding functionals in the case of the Shannon entropy, the Shannon cross-entropy and the Kullback-Leibler divergence. In this work we deduce a sharp inequality relating the one-parameter families of Rényi entropies, Rényi divergences and Rényi cross entropies as a direct consequence of Jensen inequality. More precisely, we show that the sum of the Rényi entropy and Rényi divergence is bounded by the Rényi cross entropy when the corresponding three entropic parameters satisfy a precise and simple algebraic relation. Remarkably, the equality is reached in the case that one of the probability densities is the escort density (for the definition see for example~\cite{Abe2003}) of the other. We think that this fact could be of potential theoretical and applied interest, as it plays a key role in the nonextensive formalism of statistical physics~\cite{Bercher2012(divergences),Tsallis2009(book)}.

Besides the previously mentioned informational functionals, some other functionals have played traditionally a fundamental role in the development of information theory: the Fisher information measures and the absolute moments. Introduced in ~\cite{Fisher1925}, the Fisher information of a derivable probability density function has been later extended in different ways to the relative framework, see for example~\cite{Hammad1978,Antolin2009,Martin2013,Antolin2014,Yamano2021}. As a parallel way of improving the theory, biparametric extensions have been established in the literature for the Fisher information, as well as for relative Fisher-like measures \cite{Lutwak2005,Bercher2012}, among other measures in a more general framework~\cite{Toranzo2018,Tempesta2011,Tempesta2019}. Very recently, a new relative functional of Fisher type has been proposed by the authors~\cite{IPT2025}. This new functional has the remarkable property of scaling invariance, and this property has allowed for establishing sharp informational inequalities involving the Kullback-Leibler and Rényi divergences. In an alternative direction, relative versions of the $p$-th moments have also been proposed in ~\cite{IPT2025}. \textit{Cross}-counterparts of some of  these functionals will be introduced in this work. 

More sophisticated functionals than the above quoted ones, depending on more than two densities, have been defined in the literature by using the Jensen-Shannon divergence~\cite{Lin2002}, Jensen-Fisher divergence~\cite{Sanchez-Moreno2012}, or Bergman distances~\cite{Stummer-Vajda2012}. In this direction, we also propose a new functional depending on three probability density functions completing the above mentioned structure of relative functionals. With the aid of this new functional, we establish a sharp inequality  bounding the difference between two Rényi divergences. 

In recent years, certain types of probability-preserving transformations have been successfully employed in transporting informational inequalities. On the one hand, the so-called differential escort transformation has been used to extend the moment-entropy, Stam and Cramér-Rao inequalities by transporting their biparametric counterparts introduced in~\cite{Lutwak2005} to new functionals~\cite{Zozor2017,Puertas2025}. On the other hand, a pair of mutually inverse transformations, called \textit{up} and \textit{down} transformations, have allowed to highlight a mirrored domain of validity for the entropic parameters in the above mentioned inequalities~\cite{IP2025}, but also to extend them to more general classes of informational functionals depending either on the second derivative or on some incomplete weighted integrals of the density~\cite{IP2025(b),IP2025(c)}. All of these transformations, which play a central role in this work, can be seen as particular cases of the generalized Sundman-like transformations recently defined and employed in the analysis of certain classes of differential equations related to some notable non-linear version of the Schrödinger equation~\cite{GPPT25(b)}. In addition, a relative counterpart of the differential-escort transformation has been defined in \cite{IPT2025}. The latter transformation has thus motivated the definition of the previously discussed Fisher-like and moment-like scale-invariant relative measures and prove their corresponding informational inequalities. These transformations play a key role to extend the inequalities involving two densities to inequalities involving three densities.

The present paper has, in our opinion, two strong points that we want to highlight here. As a first result, as already mentioned in the previous paragraphs, we establish a sharp inequality involving the Rényi entropy, the Rényi divergence and the Rényi cross-entropy. The other main result of this work is the establishment of a general framework consisting in introducing a pair of measure-preserving transformations in such a way that the pair of transformed and reciprocally transformed densities preserves the Rényi divergence. We then systematically employ this general framework together with many of the above mentioned transformations to derive new sharp inequalities. As a general byproduct of the present results, we emphasize the establishment of sharp bounds for the Rényi divergence in terms of functionals of different nature (such as moments, entropies, Fisher information measures and beyond) and their cross-like counterparts. The cases of equality are explicitly given in all these new inequalities, proving their sharpness.

We believe that the inequalities given in Section 3 are only a few examples among the ones that can be obtained by employing the general framework. We thus leave to the interested reader to derive more sharp bounds by applying the same framework in connection with other measure-preserving transformations.


\section{An inequality involving the R\'enyi entropy, the R\'enyi divergence and the R\'enyi cross-entropy}

The aim of this section is to state and prove a new inequality involving some already well-established informational functionals. As explained in the Introduction, this inequality is the root of the forthcoming applications given in the rest of the paper.

Throughout this paper we consider the following general framework: $\alpha,\beta,\gamma\in\mathbb R\setminus\{1\}$ are three real numbers satisfying the following relation
\begin{equation}\label{eq:relation}
(\alpha-\beta)(\alpha-\gamma)=(\alpha-1)^2.
\end{equation}
Moreover, $f$, $g$ and $h$ will be probability density functions such that
\begin{equation}\label{cond:support}
	\begin{split}
		&{\rm supp}\,f={\rm supp}\,g={\rm supp}\,h=\overline{\Omega}, \quad \Omega=(c,d)\subseteq\Rset, \\
		&f(x)>0, \ g(x)>0, \ h(x)>0, \quad {\rm for \ any} \ x\in\Omega,
	\end{split}
\end{equation}
where $\Omega$ can be either bounded or unbounded and $\overline{\Omega}$ denotes the closure of the set $\Omega$. We also employ throughout the paper the notation $g(x)\propto f(x)$ to say that the functions $g$ and $f$ are proportional; that is, $g(x)/f(x)$ is a constant.

\subsection{A brief review on informational measures}

Before stating the main inequality, and for the sake of completeness, we recall below the definitions of the three informational functionals participating in it.

\medskip

\noindent \textbf{R\'enyi entropy.} Given $\alpha\neq1$, the differential R\'enyi entropy of $\alpha$-order of a probability density function $\pdf$ is defined as
\begin{equation*}
	R_\alpha[\pdf] = \frac1{1-\alpha} \log\left( \int_\Rset [\pdf(x)]^\alpha \, dx \right).
\end{equation*}
In the limiting case $\alpha=1$ we recover the well-known differential Shannon entropy
\begin{equation*}
	\lim\limits_{\alpha\to 1}R_{\alpha}[\pdf]=S[\pdf]=-\int_\Rset \pdf(x) \, \log\pdf(x) \, dx.
\end{equation*}

\medskip

\noindent \textbf{R\'enyi divergence.} Given $\alpha\neq1$, the differential R\'enyi divergence of $\alpha$-order of a pair of probability density functions $\pdf$ and $g$ satisfying \eqref{cond:support} is defined as
\begin{equation*}
	D_\alpha[\pdf||g] = \frac1{\alpha-1} \log\left( \int_\Rset [\pdf(x)]^{\alpha}[g(x)]^{1-\alpha} \, dx \right).
\end{equation*}
In the limiting case $\alpha=1$ we recover the well-known Kullback-Leibler divergence
\begin{equation*}
	\lim\limits_{\alpha\to 1}D_{\alpha}[\pdf||g]=D[\pdf||g]=\int_\Rset \pdf(x) \, \log\frac{\pdf(x)}{g(x)} \, dx.
\end{equation*}

\medskip

\noindent \textbf{R\'enyi cross-entropy.} Given $\alpha\neq1$, the differential R\'enyi cross-entropy of $\alpha$-order of a probability density function $\pdf$ relative to a probability density function $g$ is defined as
\begin{equation*}
	H_\alpha[\pdf;g] = \frac1{1-\alpha} \log\left( \int_\Rset \pdf(x)[g(x)]^{\alpha-1} \, dx \right).
\end{equation*}
In the limiting case $\alpha=1$ we recover the well-known differential Shannon cross-entropy
\begin{equation*}
	\lim\limits_{\alpha\to 1}H_{\alpha}[\pdf]=H[\pdf;g]=-\int_\Rset \pdf(x) \, \log g(x) \, dx.
\end{equation*}
Note that $H_{\alpha}[f;f]=R_{\alpha}[f]$.

\subsection{The main inequality}

We are now in a position to state and prove the inequality representing the starting point of the applications presented in the rest of the paper.
\begin{theorem}\label{theoremRRR}
	 Let $\alpha,\beta,\gamma$ be three real numbers satisfying Eq.~\eqref{eq:relation}. If $\alpha>\beta$ then
\begin{equation}\label{eq:ineqRRR}
R_\alpha[f]+D_\beta[f||g]\leqslant H_{\gamma}[f;g].
\end{equation}
In the opposite case $\alpha<\beta$ the inequality is reversed. Moreover, the equality holds if and only if
\begin{equation}\label{eq:main}
g(x)\propto [f(x)]^{\frac{\beta-1}{\beta-\alpha}}.
\end{equation}
\end{theorem}
\begin{proof}
The proof is an application of Jensen's inequality. Indeed, it follows from Jensen's inequality that
	\begin{equation}\label{ineq:Jensen}
		\left(\frac{\int_{\mathbb R}\chi(x)\xi(x)\,dx}{\int_{\mathbb R}\xi(x)\,dx}\right)^K\leqslant \frac{\int_{\mathbb R} [\chi(x)]^K \xi(x)\,dx}{\int_{\mathbb R}\xi(x)\,dx}
	\end{equation}
for any integrable functions $\chi$ and $\xi$, when $K>1$ or $K<0$ (and the opossite inequality for $0<K<1$). Choosing in \eqref{ineq:Jensen}
$$
\chi(x)=[g(x)]^{A}[f(x)]^B \quad {\rm and}\quad \xi(x)=[f(x)]^C
$$
one obtains
	\begin{equation}\label{auxineq}
		\left(\frac{\int_{\mathbb R} [g(x)]^{A}[f(x)]^{B+C}\,dx}{\int_{\mathbb R}[f(x)]^C\,dx}\right)^K\leqslant \frac{\int_{\mathbb R} [g(x)]^{AK}[f(x)]^{BK+C}\,dx}{\int_{\mathbb R}[f(x)]^C\,dx}.
	\end{equation}
Now, if we further particularize \eqref{auxineq} by setting
$$
B=1-\alpha ,\quad  C=\alpha ,\quad AK=1-\beta ,\quad BK+C=\beta ,
$$
or equivalently,
\begin{equation}\label{eqdem:cts}
A=\frac{(1-\beta )(1-\alpha )}{\beta -\alpha },\quad B=1-\alpha ,\quad  C=\alpha ,\quad K=\frac{\beta -\alpha }{1-\alpha },
\end{equation}
we find
\begin{equation}\label{ineq:aux2}
	\left(\int_\mathbb R	\left[g(x)\right]^{\frac{(1-\beta)(1-\alpha )}{\beta -\alpha }}\,f(x)dx \right)^{\frac{\beta -\alpha }{1-\alpha }}\left(\int_{\mathbb R}[f(x)]^\alpha \,dx\right)^{\frac{1-\beta }{1-\alpha }}\leqslant \int_{\mathbb R} [g(x)]^{1-\beta }[f(x)]^{\beta }\,dx.
\end{equation}
In the following cases
\begin{equation}\label{eq:cases}
	\beta <1\;\; {\rm and}\;\; K\notin[0,1],\qquad {\rm as\ well\ as},\qquad \beta >1\;\;{\rm and}\;\;K\in(0,1),
\end{equation}
the inequality \eqref{ineq:aux2} can be rewritten as
\begin{equation}\label{eq:ineqRRR2}
	\left(\int_\mathbb R	\left[g(x)\right]^{\frac{(1-\beta )(1-\alpha )}{\beta -\alpha}}\,f(x)dx \right)^{\frac{\beta -\alpha }{(1-\alpha )(1-\beta )}}\left(\int_{\mathbb R}[f(x)]^\alpha \,dx\right)^{\frac{1}{1-\alpha }}\leqslant \left(\int_{\mathbb R} [g(x)]^{1-\beta }[f(x)]^{\beta }\,dx\right)^\frac1{1-\beta }.
\end{equation}
By taking logarithms in the previous inequality and noting that Eq.~\eqref{eq:relation} implies
	$$
	\frac{(1-\beta )(1-\alpha )}{\beta -\alpha }=\gamma-1,
	$$
we deduce that
\begin{equation*}
-H_{\gamma}[f;g]+R_\alpha [f]\leqslant -D_\beta [f||g],
\end{equation*}
which is obviously equivalent to the inequality~\eqref{eq:ineqRRR}. We are left to describe the conditions which the parameters $\alpha,\beta$ must fulfill. On the one hand, the condition $K>1$ is equivalent to
$$
K-1=\frac{\beta -1}{1-\alpha }>0.
$$
Recalling from \eqref{eq:cases} that $K>1$ implies $\beta<1$, we have $\alpha>1$. On the other hand, the opposite case $K<0$ necessarily implies $\alpha<1$, since again \eqref{eq:cases} entails that $\beta<1$ and thus, if $\alpha>1$, both $\beta-\alpha$ and $1-\alpha$ would be negative. Once established that $\alpha<1$, it immediately follows that $\beta<\alpha$ from the negativity of $K$. Both cases $K>1$ and $K<0$, with $\beta <1$, can be summarized in the condition $\beta <\min\{1,\alpha\}$. In the remaining case $\beta>1$ and $K\in(0,1)$, one finds that $K<1$ implies $\alpha>1$, which in turn implies $\alpha>\beta>1$ after imposing the condition $K>0$. Thus, in both cases $\beta>1$ and $\beta<1$ we have reached the same condition $\beta<\alpha$. Finally, the equality in \eqref{eq:ineqRRR} is achieved when
$$
g(x)\propto [f(x)]^{-\frac AB}=[f(x)]^{\frac{\beta-1}{\beta-\alpha}},
$$
completing the proof.
\end{proof}
\noindent\textbf{Remark.} Note that, in the limiting case $\alpha=\beta=\gamma=1$, we obtain a well-known and very easy identity,
$$
S[f]+D[f||g]=H[f;g],
$$
which can be checked by direct calculation from the definitions. Moreover, let us observe that the inequality~\eqref{eq:ineqRRR2} is also valid for more general functions $f$ and $g$ (dropping the hypothesis of necessarily being probability density functions) whenever the involved integral are finite.

It is worth mentioning that the equality is reached when $g$ is an escort transformation of $f$ in the non-extensive formalism~\cite{Abe2003,Beck2004,Bercher2011}, that is
$$
g(x)=\frac{[f(x)]^{\frac{\beta-1}{\beta-\alpha}}}{\int_\mathbb R [f(t)]^{\frac{\beta-1}{\beta-\alpha}}\,dt}.
$$

\section{Applications to further inequalities}

In this section, we employ the inequality \eqref{eq:ineqRRR} together with a number of measure-preserving transformations in order to obtain new sharp inequalities connecting functionals of interest in Information Theory. Some of these inequalities connect some simple and already well studied functionals, and we expect them to be a starting point for further applications.

\subsection{General framework}\label{subsec:gf}

Let $f$ be a probability density and let $\widetilde f$ be its transformed density through certain measure-preserving transformation $\mathcal{O}$. More precisely, we define
$$
\widetilde f(y)=\mathcal O[f(x)],\quad \widetilde f(y)dy=f(x)dx,
$$
or equivalently
$$
\widetilde f(y)=\mathcal O[f(x(y))],\quad y'(x)=\frac{f(x)}{\mathcal O[f(x)]}.
$$
In order to keep the notation as simple as possible, we employ the simplified notation $\mathcal O[f]$, but we stress here that, in the most general case, the transformation might depend not only on the proper probability density, but also on its derivative function and the variable, that is, $\mathcal O[x,f(x),f'(x)].$ Fixing the density $f$ and the transformation $\mathcal O$ as above, and given a probability density $g$, one can define the following transformation, that will be called \emph{the reciprocal transformation} to $\mathcal{O}$:
\begin{equation}\label{eq:par_transf}
\overline{\mathcal O}[g(x)]=\frac{g(x)}{f(x)}\mathcal O[f(x)],\qquad y'(x)=\frac{f(x)}{\mathcal O[f(x)]}.
\end{equation}
It is obvious that $\overline{\mathcal{O}}$ is a measure-preserving transformation as well.
Moreover, we have the following easy but fundamental property:
\begin{proposition}\label{prop:div}
In the previous notation and definitions, we have
\begin{equation}\label{eq:div}
	D_\gamma[\mathcal O[f]||\overline{\mathcal O}[g]]=D_\gamma[f||g]
\end{equation}
for any $\gamma\in\Rset$.
\end{proposition}
\begin{proof}
We deduce from the definitions of $\mathcal{O}$, $\overline{\mathcal{O}}$ and of the R\'enyi divergence that
\begin{equation*}
\begin{split}
D_\gamma[\mathcal O[f]||\overline{\mathcal O}[g]]&=\frac{1}{\gamma-1}\log\left(\int_{\Rset}\mathcal{O}[f]^{\gamma}(y)\overline{\mathcal{O}}[g]^{1-\gamma}(y)\,dy\right)\\
&=\frac{1}{\gamma-1}\log\left[\int_{\Rset}\left(\frac{\mathcal{O}[f](y)}{\overline{\mathcal{O}}[g](y)}\right)^{\gamma}\overline{\mathcal{O}}[g](y)\,dy\right]\\
&=\frac{1}{\gamma-1}\log\int_{\Rset}\left(\frac{f(x)}{g(x)}\right)^{\gamma}g(x)\,dx=D_{\gamma}[f||g].
\end{split}
\end{equation*}
\end{proof}
In the next subsections, we apply this general framework to some recently introduced measure-preserving transformations and obtain new inequalities by transporting the inequality~\eqref{eq:ineqRRR}, all them involving, as a consequence of Proposition \ref{prop:div}, the R\'enyi divergence.

\subsection{Diferential-escort transformation}

The first transformation that we shall choose in place of $\mathcal{O}$ in the general framework is the differential-escort transformation, introduced in \cite{Zozor2017, Puertas2019}, which we recall here. If $f$ is a probability density function and $\xi\in\Rset$, then the differential-escort transformation is defined as
\begin{equation}\label{def:escort}
\mathfrak{E}_{\xi}[f](y)=[f(x(y))]^{\xi}, \quad y'(x)=[f(x)]^{1-\xi}.
\end{equation}
We next define the reciprocal transformation to the differential-escort one.
\begin{definition}
Let $\xi\in\mathbb R$ be a real number. Let $\pdf$ be a probability density, and $f_\xi=\mathfrak E_\xi[\pdf]$ its differential-escort transformation. We define the reciprocal transformation $\overline{\mathfrak E}_\xi$ by
\begin{equation}\label{def:rec_escort}
	\overline{\mathfrak E}_\xi[g](y)=\frac{g(x(y))}{f(x(y))}\mathfrak E_\xi[f](y)=g(x(y))[f(x(y))]^{\xi-1},\quad y' (x)=[f(x)]^{1-\xi}.
\end{equation}
\end{definition}
Some basic properties of the transformation $\overline{\mathfrak{E}}$ are listed below.
\begin{lemma}\label{lemma1}
Let $\gamma,\xi$ be real numbers and $f,g$ probability densities. Then
	\begin{equation}\label{eqlem:div}
		D_\gamma\left[\mathfrak E_\xi[f]||\overline{\mathfrak E}_\xi[g]\right]=D_\gamma[f||g].
	\end{equation}
and
	\begin{equation}\label{eqlem:cross}
		H_\gamma\left[\mathfrak E_\xi[f]\,;\,\overline{\mathfrak E}_\xi[g]\right]=H_{\gamma,\xi}[f;g].
	\end{equation}
where
	\begin{equation*}
		H_{\gamma,\xi}[f;g]:=\int_\mathbb{R} [f(x)]^{1+(\xi-1)(\gamma-1)}[g(x)]^{\gamma-1}\,dx.
	\end{equation*}
\end{lemma}
\begin{proof}
The identity \eqref{eqlem:div} is a particular case of the equality \eqref{eq:div} established in Proposition \ref{prop:div}. We next compute the following integral
\begin{equation*}
\begin{split}
H_\gamma\left[\mathfrak E_\xi[f]\,;\,\overline{\mathfrak E}_\xi[g]\right]&=\int_\mathbb R\mathfrak E_\xi[f](y)\left[\overline{\mathfrak E}_\xi[g](y)\right]^{\gamma-1}dy\\
&=\int_\mathbb R  f(x)\left[g(x)f(x)^{\xi-1}\right]^{\gamma-1}dx\\
&=\int_\mathbb R [f(x)]^{1+(\xi-1)(\gamma-1)}\left[g(x)\right]^{\gamma-1}dx=H_{\gamma,\xi}[f;g],
\end{split}
\end{equation*}
proving \eqref{eqlem:cross}.
\end{proof}
The following inequality is obtained by combining the general framework applied to the transformations $\mathfrak{E}$ and $\overline{\mathfrak{E}}$ with the inequality \eqref{eq:ineqRRR}.
\begin{theorem}\label{th:dif-escort}
Let $\alpha,\beta,\gamma\in\mathbb R\setminus\{1\}$ be three real numbers satisfying \eqref{eq:relation} and let $f$, $g$ be two probability density functions satisfying \eqref{cond:support}. Then, if $\alpha>\beta$, for any $\xi\in\Rset$, we have
\begin{equation}\label{ineq:dif-escort}
\xi R_{1+(\alpha-1)\xi}[f]+D_\beta[f||g]\leqslant H_{\gamma,\xi}[f;g],
\end{equation}
while if $\alpha<\beta$, the inequality is reversed. The inequality \eqref{ineq:dif-escort} is sharp and the equality is achieved if and only if
\begin{equation}\label{eq:ineq-dif-escort}
g(x)\propto f(x)^{1+\frac{\xi(1-\alpha)}{\alpha-\beta}}.
\end{equation}
\end{theorem}
\begin{proof}
Assume that $\alpha>\beta$. We apply the inequality \eqref{eq:ineqRRR} to the transformed probability density functions $\mathfrak E_\xi[f]$ and $\overline{\mathfrak E}[g]$ to obtain
\begin{equation}\label{eq:interm1}
			R_\alpha[\mathfrak E_\xi[f]]+D_\beta[\mathfrak E_\xi[f]||\overline{\mathfrak E}_\xi[g]]\leqslant H_{\gamma}[\mathfrak E_\xi[f];\overline{\mathfrak E}_\xi[g]].
\end{equation}
Recalling that (see \cite{Puertas2019})
$$
R_\alpha[\mathfrak E_\xi[f]]=\xi R_{1+(\alpha-1)\xi}[f],
$$
the inequality \eqref{ineq:dif-escort} readily follows as a consequence \eqref{eq:interm1}, \eqref{eqlem:div} and \eqref{eqlem:cross}. It is obvious that the inequality \eqref{ineq:dif-escort} is reversed for $\alpha<\beta$, since the inequality sign is inherited from \eqref{eq:ineqRRR}. The equality is achieved when $\mathfrak{E}_{\xi}[f]$ and $\overline{\mathfrak{E}}_{\xi}[g]$ satisfy \eqref{eq:main}, that is,
$$
g(x)f(x)^{\xi-1}=\overline{\mathfrak{E}}_{\xi}[g]\propto(\mathfrak{E}_{\xi}[f])^{\frac{1-\beta}{\alpha-\beta}}=f(x)^{\frac{\xi(1-\beta)}{\alpha-\beta}},
$$
which leads to \eqref{eq:ineq-dif-escort}.
\end{proof}

Let us observe that the inequality~\eqref{eq:ineqRRR2} essentially preserves the \textit{minimizing relation} between $f$ and $g$, that is, $g$ being a escort density of $f$.

\subsection{Relative differential-escort transformation}

The next transformation employed is the recently introduced relative differential-escort transformation, see \cite{IPT2025}. We first recall its definition here for the sake of completeness.
\begin{definition}
Let $f$ and $h$ be two probability density functions satisfying \eqref{cond:support} and $\xi\in\Rset$. We define the relative differential-escort transformed density of $\xi$-order of $f$ as
\begin{equation}\label{def:rel-escort}
\mathfrak R_\xi^{h}[\pdf](y):=\left(\frac{\pdf(x(y))}{h(x(y))}\right)^\xi,\quad y'(x)=\pdf(x)^{1-\xi}h(x)^{\xi}.
\end{equation}
\end{definition}
We next introduce the reciprocal transformation, according to the general framework given in Section \ref{subsec:gf}.
\begin{definition}
Given three probability density functions $f$, $g$, $h$ satisfying \eqref{cond:support} and a real parameter $\xi$, we define the transformed density of $g$, for fixed probability densities $f$ and $h$, as:
\begin{equation}\label{def:rec-rel-escort}
\overline{\mathfrak R^{h}_\xi}[g](y)=\frac{g(x)}{f(x)}\left(\frac{f(x)}{h(x)}\right)^{\xi},\qquad y'(x)=[f(x)]^{1-\xi}[h(x)]^\xi.
\end{equation}
\end{definition}
In order to state the next inequality, we first introduce a new informational functional that has an interesting expression, combining in some sense (as indicated below) the properties of a cross-entropy and of a divergence in some particular cases. This is why, we decided to give this functional the name of \emph{cross-divergence}.
\begin{definition}[Cross-divergence]
Let $f,g,h$ three probability density functions satisfying \eqref{cond:support}, and let $a,b$ be two real numbers. The cross-divergence of $(a,b)$-order of the functions $f$ and $g$ with reference function $h$ is defined as
\begin{equation}\label{def:cross-div}
	\widetilde{H}_{a,b}[f;g||h]=\frac{1}{1-a}\log\int_\mathbb{R} f(x)\left(\frac{[f(x)]^{b-1} g(x)}{[h(x)]^b}\right)^{a-1}\,dx.
\end{equation}
In the particular case $b=1$, we will denote
\begin{equation*}
\widetilde{H}_{a}[f;g||h]=\widetilde{H}_{a,1}[f;g||h]=\frac{1}{1-a}\log\int_\mathbb{R} f(x)\left(\frac{g(x)}{h(x)}\right)^{a-1}\,dx.
\end{equation*}
\end{definition}

\noindent \textbf{Remark. Particular cases.} Note that the functional $\widetilde{H}_{a,b}[f;g||h]$ reduces to a divergence whenever $f=g,\,g=h$ or $f=h$. Indeed, when $f=g$ we obtain
$$
\widetilde{H}_{a,b}[f;f||h]=-bD_{1+b(a-1)}[f||h],
$$
while if $g=h$ then
$$
\widetilde{H}_{a,b}[f;h||h]=(1-b)D_{1+(a-1)(b-1)}[f||h].
$$
In particular, if $b=1$ we get $\widetilde{H}_{a}[f;h||h]=0$. Finally, if $f=h$ the influence of the parameter $b$ disappears and we obtain
$$
\widetilde{H}_{a,b}[f;g||f]=D_{2-a}[f||g].
$$
Several other interesting particular cases are listed below, for any pair of probability density functions $f$ and $g$:
\begin{itemize}
	\item When $b=0$, one has $$\widetilde{H}_{a,0}[f;g||h]=D_{2-a}[f||g].$$
	\item When $a=2$ follows
	$$
\widetilde H_{2,b}[f;g||h]=-\log \int_\mathbb R g(x)\left(\frac{f(x)}{h(x)}\right)^b \,dx=b \widetilde H_{b+1}[g;f||h].
	$$
	\item When $(1-a)b=1$ follows
	$$
	\widetilde H_{a,b}[f;g||h]=\frac1{1-a}\log \int_\mathbb R h(x)\left(\frac{g(x)}{f(x)}\right)^{a-1} \,dx=\widetilde H_{a}[h;g||f].
	$$
	\item Letting $\overline a $ and $\overline b$ such that
	$$b(1-a)=\overline a-1\quad {\rm and}\quad a-1=\overline b(1-\overline a),
	$$ or equivalently,
	$$\overline a=1+b(1-a),\quad {\rm and}\quad \overline b=1/b,
	$$
we have
	\begin{eqnarray*}
	\widetilde H_{a,b}[f;g||h]&=&\frac1{1-a}\log \int_\mathbb R [f(x)]^{1+(a-1)(b-1)}[g(x)]^{a-1}[h(x)]^{b(1-a)} \,dx
	\\
	&=&
	\frac1{1-a}\log \int_\mathbb R [f(x)]^{1+(\overline a-1)(\overline b-1)}[g(x)]^{\overline b(1-\overline a)}[h(x)]^{\overline a-1} \,dx
	\\
	&=&
	-b \widetilde H_{\overline a,\overline b}[f;h||g].
	\end{eqnarray*}
\end{itemize}

We now state and prove an inequality relating the R\'enyi divergence and the cross-divergence.
\begin{theorem}
Let $f,g,h$ be three probability density functions satisfying \eqref{cond:support} and let $\alpha,\beta,\gamma$ be three real numbers satisfying the relation \eqref{eq:relation}. Let $\xi$ be a real number. If $\alpha>\beta$, then
\begin{equation}\label{ineq:rel-dif-escort}
	D_\beta[f||g]-\xi D_{1+(\alpha-1)\xi}[f||h]\leqslant \widetilde H_{\gamma,\xi}[f;g||h].
\end{equation}
In the particular case $\xi=1$ one obtains
\begin{equation*}
D_\beta\left[f||g\right]- D_{\alpha}[f||h]\leqslant \widetilde H_{\gamma}[f;g||h].
\end{equation*}
When $\alpha<\beta$ the previous inequalities are reversed. The equality in \eqref{ineq:rel-dif-escort} is achieved if and only if
\begin{equation}\label{eq:ineq-rel-dif-escort}
g(x)\propto f(x)\left(\frac{f(x)}{h(x)}\right)^{\frac{\xi(1-\alpha)}{\alpha-\beta}}.
\end{equation}
\end{theorem}
\begin{proof}
Assume that $\alpha>\beta$. The inequality~\eqref{eq:ineqRRR} applied to the pair of transformed densities $\mathfrak{R}^{h}_\xi[f]$ and $\overline{\mathfrak{R}^{h}_\xi}[g]$ yields
\begin{equation}\label{eq:interm2}
R_\alpha[\mathfrak{R}^{h}_\xi[f]]+D_\beta\left[\mathfrak{R}^{h}_\xi[f]\,\big|\big|\,\overline{\mathfrak{R}^{h}_\xi}[g]\right]\leqslant H_\gamma[\mathfrak{R}^{h}_\xi[f];\overline{\mathfrak{R}^{h}_\xi}[g]].
\end{equation}
The fact that
$$
D_\beta\left[\mathfrak{R}^{h}_\xi[f]\,\big|\big|\,\overline{\mathfrak{R}^{h}_\xi}[g]\right]=D_\beta[f||g]
$$
is a particular case of the equality \eqref{eq:div}. For the first term in \eqref{eq:interm2}, we recall from \cite[Lemma 3.2]{IPT2025} (after taking logarithms in the equality therein) that
$$
R_\alpha[\mathfrak{R}^{h}_\xi[f]]=-\xi D_{1+(\alpha-1)\xi}[f||h].
$$
Finally,
\begin{eqnarray*}
	H_\gamma[\mathfrak{R}^{h}_\xi[f];\overline{\mathfrak{R}^{h}_\xi}[g]]&=&\frac1{1-\gamma}\log\int_\mathbb R  \mathfrak{R}^{h}_\xi[f](y)\left(\overline{\mathfrak{R}^{h}_\xi}[g](y)\right)^{\gamma-1}\,dy
	\\
	&=&
	\frac1{1-\gamma}\log\int_\mathbb R  f(x)\left([f(x)]^{\xi-1}\,g(x)\,[h(x)]^{-\xi}\right)^{\gamma-1}\,dx
	\\
	&=&
	\frac1{1-\gamma}\log\int_\mathbb R  [f(x)]^{1+(\gamma-1)(\xi-1)}[g(x)]^{\gamma-1}[h(x)]^{\xi(1-\gamma)}\,dx\\
&=& \widetilde{H}_{\gamma,\xi}[f;g||h].
\end{eqnarray*}
The inequality \eqref{ineq:rel-dif-escort} follows easily by replacing the previous identities in \eqref{eq:interm2}. It is obvious that the inequality sign is reversed if $\alpha<\beta$, since the inequality is inherited from \eqref{eq:ineqRRR}. The equality in \eqref{ineq:rel-dif-escort} is achieved when, according to \eqref{eq:main}, the following proportionality holds true:
$$
\frac{g(x)}{f(x)}\left(\frac{f(x)}{h(x)}\right)^{\xi}=\overline{\mathfrak{R}^{h}_\xi}[g]\propto\mathfrak{R}^{h}_\xi[f]^{\frac{1-\beta}{\alpha-\beta}}
=\left(\frac{f(x)}{h(x)}\right)^{\frac{\xi(1-\beta)}{\alpha-\beta}},
$$
which gives \eqref{eq:ineq-rel-dif-escort}, completing the proof.
\end{proof}

\subsection{Biparametric down transformation}

The next transformation that we employ as a particular case of the general framework given in Section \ref{subsec:gf} is the biparametric down transformation. This transformation has been introduced and thoroughly studied by the authors in connection with new informational functionals in the recent work \cite{IP2025(c)} and we recall its definition below.
\begin{definition}\label{def:bip_down}
Let $f:(c,d)\mapsto\Rset$ be a derivable probability density function such that $f'(x)<0$ for any $x\in(c,d)$, where $-\infty<c<d\leqslant\infty$. For $a$, $b\in\Rset$, we define the \emph{biparametric down transformation} by
\begin{equation}\label{eq:bip_down}
\mathfrak{D}_{a,b}[f](y):=\frac{f(x)^{a}}{|f'(x)|^{b}}, \quad y'(x)=f(x)^{1-a}|f'(x)|^{b}.
\end{equation}
\end{definition}
Let us recall here that, for $b=1$, the down transformation $\mathfrak{D}_{a,1}\equiv\mathfrak{D}_{a}$ had been previously defined and its applications studied in \cite{IP2025, IP2025(b)}. Starting from this definition, we can introduce the reciprocal transform $\overline{\mathfrak{D}}_{a,b}$ as a particular case of \eqref{eq:par_transf} adapted to the transformation $\mathfrak{D}_{a,b}$.
\begin{definition}\label{def:rec-down}
Let $a,b$ be real numbers, and let $f$ and $g$ be two probability density functions such that $f$ is decreasing and derivable. We define the following transformation
\begin{equation}\label{eq:rec-down}
	\overline{\mathfrak D}_{a,b}[g](y)=\frac{g(x)}{f(x)}\frac{[f(x)]^a}{|f'(x)|^b}, \qquad y(x)=[f(x)]^{1-a}|f' (x)|^{b}.
\end{equation}
\end{definition}

The following informational functional, named generalized Fisher information, was introduced by Lutwak and Bercher~\cite{Lutwak2005, Bercher2012, Bercher2012a} and it acts on derivable probability density functions.
	\begin{definition}[Generalized Fisher information]
		Given $p>1$ and $\lambda \in \Rset^*$, the $(p,\lambda)$-Fisher information of a probability density function $f$ is defined as
\begin{equation}\label{eq:def_FI}
	\phi_{p,\lambda}[\pdf]=\big(F_{p,\lambda}[\pdf]\big)^{\frac{1}{p\lambda}},\quad F_{p,\lambda}[\pdf]=\int_\Rset [f(x)]^{1+p(\lambda-2)} \left|\frac{d\pdf}{dx}(x)\right|^{p}  \, dx
\end{equation}
whenever $\pdf$ is differentiable on the closure of its support. In particular, the standard Fisher information is recovered as the $(2,1)$-Fisher information.
\end{definition}
Before stating a similar sharp inequality as in the previous sections, we also introduce the following informational functional, which we have called the \emph{generalized cross-Fisher information}.
\begin{definition}[Generalized cross-Fisher information]\label{def:cross-Fisher}
Let $a,b,c$ be real numbers, and $f$ a derivable probability density. The $(a,b,c)$-cross-Fisher information of $f$ relative to a probability density $g$ is defined as
\begin{equation}
\begin{split}
	\phi^{\rm (cr)}_{a,b,c}[f;g]:&=\left(\int_{\mathbb R}[f(x)]^{1+(a-1)c} [g(x)]^{-c}|f' (x)|^{bc}\right)^{\frac1{c}}\\
&=\left(\int_{\mathbb R}[f(x)]^{1+(a-2)c} \left(\frac{f(x)}{g(x)}\right)^{c}|f' (x)|^{bc}\right)^{\frac1{c}}.
\end{split}
\end{equation}
\end{definition}
Note that, in the particular case $b=1$ and $f=g$, $\phi^{\rm (cr)}_{a,1,c}[f;f]$ reduces to the standard biparametric Fisher information (see for example \cite{Zozor2017}). This equality justifies the given name of cross-Fisher information.
\begin{theorem}\label{th:bip-down}
Let $\alpha$, $\beta$, $\gamma$ be three real numbers satisfying \eqref{eq:relation} and $f$, $g$ be two probability density functions satisfying \eqref{cond:support} such that $f$ is decreasing and derivable. Then, for any $a$, $b\in\Rset$ such that $b\neq0$ and $a\neq2b$, we have the following inequality:
\begin{equation}\label{ineq:bip-down}
\phi_{(1-\alpha)b,2-\frac ab}^{2b-a}[f]e^{D_\beta[f||g]}\leqslant\phi^{\rm (cr)}_{2-a,b,1-\gamma}[f;g].
\end{equation}
The equality in \eqref{ineq:bip-down} is attained if and only if
\begin{equation}\label{eq:ineq-bip-down}
g(x)\propto f(x)^{A}|f'(x)|^{B}, \quad A=1+\frac{a(1-\alpha)}{\alpha-\beta}, \quad B=\frac{b(\alpha-1)}{\alpha-\beta}.
\end{equation}
\end{theorem}
\begin{proof}
Starting from inequality~\eqref{eq:ineqRRR} applied to $\mathfrak D_{a,b}[f]$ and $\overline{\mathfrak D}_{a,b}[g]$, we obtain
\begin{equation}\label{eq:interm3}
R_\alpha[\mathfrak D_{a,b}[f]]+D_\beta[\mathfrak D_{a,b}[f]||\overline{\mathfrak D}_{a,b}[g]]\leqslant H_{\gamma}[\mathfrak D_{a,b}[f];\overline{\mathfrak D}_{a,b}[g]].
\end{equation}
Since it is assumed that $b\neq0$ and $a\neq2b$, we recall from \cite[Lemma 3.1]{IP2025(c)} that
$$
e^{R_\alpha[\mathfrak D_{a,b}[f]]}=N_{\alpha}[\mathfrak{D}_{a,b}[f]]=\phi_{(1-\alpha)b,2-\frac ab}^{2b-a}[f].
$$
Moreover, the identity
\begin{equation}\label{eq:interm5}
D_\beta[\mathfrak D_{a,b}[f]||\overline{\mathfrak D}_{a,b}[g]]=D_{\beta}[f||g],
\end{equation}
follows as a consequence of the general identity \eqref{eq:div}. Taking exponentials, Eq. \eqref{eq:interm3} reads
\begin{equation}\label{eq:interm4}
\phi_{(1-\alpha)b,2-\frac ab}^{2b-a}[f]e^{D_\beta[f||g]}\leqslant \exp\left\{H_{\gamma}[\mathfrak D_{a,b}[f];\overline{\mathfrak D}_{a,b}[g]]\right\}.
\end{equation}
It only remains to calculate the right-hand side of \eqref{eq:interm4}. We have
\begin{eqnarray*}
\exp\{H_{\gamma}[\mathfrak D_{a,b}[f];\overline{\mathfrak D}_{a,b}[g]]\}
&=&
\left[\int_\mathbb R \mathfrak D_{a,b}[f](y)\left(\overline{\mathfrak D}_{a,b}[g](y)\right)^{\gamma-1}dy\right]^{\frac{1}{1-\gamma}}
\\
&=& \left[\int_\mathbb R f(x)\left(\frac{g(x) [f(x)]^{a-1}}{|f' (x)|^b}\right)^{\gamma-1}\,dx\right]^{\frac{1}{1-\gamma}}
\\
&=& \left[\int_\mathbb R [f(x)]^{1+(a-1)(\gamma-1)} [g(x)]^{\gamma-1}|f' (x)|^{b(1-\gamma)}\,dx\right]^{\frac{1}{1-\gamma}}\\
&=& \phi^{\rm (cr)}_{2-a,b,1-\gamma}[f;g],
\end{eqnarray*}
and a substitution of the final expression in \eqref{eq:interm4} completes the proof of \eqref{ineq:bip-down}. The equality case in \eqref{ineq:bip-down} is inherited from the equality case in \eqref{eq:ineqRRR}; that is,
$$
\frac{g(x)}{f(x)}\frac{f(x)^{a}}{|f'(x)|^{b}}=\overline{\mathfrak{D}}_{a,b}[g]\propto\mathfrak{D}_{a,b}[f]^{\frac{1-\beta}{\alpha-\beta}}
=\left(\frac{f(x)^a}{|f'(x)|^b}\right)^{\frac{1-\beta}{\alpha-\beta}},
$$
that is,
$$
g(x)\propto f(x)\left(\frac{f(x)^a}{|f'(x)|^b}\right)^{\frac{1-\alpha}{\alpha-\beta}},
$$
which is obviously equivalent to \eqref{eq:ineq-bip-down}, as claimed.
\end{proof}

\noindent \textbf{Remark.} Direct calculations show that, in the case when $f$ is an exponential density, the equality holds when $g$ is also an exponential density. This also happens for $f$ and $g$ being q-exponential densities. However, when $f$ is a Gaussian density, the equality is reached for $g$ being a Rayleigh density if $\beta=\alpha+b(1-\alpha)$, or for generalized Gamma distributions for other values of $\beta$. More generally, if $f$ is a generalized normal distribution $f\propto e^{-|x|^k}$, then the equality is reached in the case that $g$ is a Weibull probability density for $\beta=\alpha+b(1-\alpha)$ or a generalized Gamma distribution in the general case. 

Taking into account the structure of informational functionals introduced in \cite{IP2025, IP2025(c)}, we can go one step below by applying once more the down transformation and derive an inequality relating the divergence $D_{\beta}[f||g]$ to the down-Fisher measure. In order to state the inequality, we first recall the definition of the down-Fisher measure.
\begin{definition}[Down-Fisher measure]\label{def:sub-fisher}
Let $f$ be a differentiable up to the second order and monotone probability density function and $p$, $q$, $\lambda$ three real numbers such that $p\neq q$. The down-Fisher measure is defined as
\begin{equation}\label{eq:sub-fisher}
\varphi_{p,q,\lambda}[\pdf]:=\int_{\Rset} f(v)^{1+p(\lambda-2)}|f'(v)|^{q}\left|\frac{p\lambda}{p-q}-\frac{\pdf(v) \pdf''(v)}{(\pdf'(v))^2}\right|^p\,dv.
\end{equation}	
\end{definition}
Together with the down-Fisher measure, we can introduce the \emph{cross-down-Fisher measure}, which is defined below:
\begin{equation}\label{eq:cross-down-Fisher}
\varphi_{a,b,c,\xi}[f;g]:=\int_{\Rset}\left(\frac{|f'(x)|^{a-b}}{f(x)^{a\xi+2b(1-\xi)}}\frac{f(x)}{g(x)}
\left|\xi-\frac{f(x)f''(x)}{(f'(x))^2}\right|\right)^{c}f(x)\,dx.
\end{equation}
With this notion in mind, we can write our next inequality.
\begin{theorem}\label{th:down-Fisher}
Let $\alpha$, $\beta$, $\gamma$ be three real numbers satisfying \eqref{eq:relation} and $f$, $g$ be two decreasing probability density functions satisfying \eqref{cond:support} and with $f$ such that
\begin{equation}\label{cond:down_twice}
	\sup\limits_{x\in\Rset}\frac{f(x)f''(x)}{(f'(x))^2}<\xi.
	\end{equation}
Then, for any $a$, $b$, $\xi\in\Rset$ such that $a\neq2b$ and $b\neq0$, we have the following inequality:
\begin{equation}\label{ineq:down-Fisher}
\varphi_{(1-\alpha)b,(1-\alpha)(a-b),\xi\left(2-\frac{a}{b}\right)}[f]^{\frac{1}{1-\alpha}}e^{D_{\beta}[f||g]}\leq \left[\varphi_{a,b,1-\gamma,\xi}^{\rm (cr)}[f;g]\right]^{\frac{1}{1-\gamma}}.
\end{equation}
The inequality \eqref{ineq:down-Fisher} is sharp and the condition for equality is given after the proof.
\end{theorem}
\begin{proof}
We start from the inequality \eqref{ineq:bip-down}, and we apply it to the transformed densities $\mathfrak{D}_{\xi}[f]$ and $\overline{\mathfrak{D}}_{\xi}[g]$, recalling that $\mathfrak{D}_{\xi}\equiv\mathfrak{D}_{\xi,1}$ in Definition \ref{def:bip_down}. The condition~\eqref{cond:down_twice} ensures that the down transformation $\down \xi$ is a decreasing function (see~\cite[Eq. 2.4]{IP2025}), as needed to apply the Theorem~\ref{th:bip-down}. On the one hand, we recall that
\begin{equation}\label{eq:interm6}
\phi_{(1-\alpha)b,2-\frac{a}{b}}^{2b-a}[\mathfrak{D}_{\xi}[f]]=\varphi_{(1-\alpha)b,(1-\alpha)(a-b),\xi\left(2-\frac{a}{b}\right)}^{\frac{1}{1-\alpha}}[f],
\end{equation}
as it follows from \cite[Lemma 3.1]{IP2025(b)}. On the other hand, we calculate the right-hand side of \eqref{ineq:bip-down} applied to the above mentioned transformed densities, taking into account the expression for the derivative of a down transformed density given in \cite[Remark 2.5]{IP2025} and the definition \eqref{eq:cross-down-Fisher} of the cross-down-Fisher measure:
\begin{equation*}
\begin{split}
\phi^{\rm (cr)}_{2-a,b,1-\gamma}&[\mathfrak{D}_{\xi}[f];\overline{\mathfrak{D}}_{\xi}[g]]=
\left[\int_\mathbb R [\mathfrak{D}_{\xi}[f](y)]^{1+(a-1)(\gamma-1)} [\overline{\mathfrak{D}}_{\xi}[g](y)]^{\gamma-1}\left|\frac{d}{dy}\mathfrak{D}_{\xi}[f](y)\right|^{b(1-\gamma)}\,dy\right]^{\frac{1}{1-\gamma}}\\
&=\left[\int_{\Rset}f(x)[\mathfrak{D}_{\xi}[f](y(x))]^{a(\gamma-1)}\left(\frac{\overline{\mathfrak{D}}_{\xi}[g]}{\mathfrak{D}_{\xi}[f]}\right)^{\gamma-1}(y(x))
\left|\frac{d}{dy}\mathfrak{D}_{\xi}[f](y(x))\right|^{b(1-\gamma)}\,dx\right]^{\frac{1}{1-\gamma}}\\
&=\left[\int_{\Rset}\left(\frac{|f'(x)|^{a-b}}{f(x)^{a\xi+2b(1-\xi)}}\frac{f(x)}{g(x)}
\left|\xi-\frac{f(x)f''(x)}{(f'(x))^2}\right|\right)^{1-\gamma}f(x)\,dx\right]^{\frac{1}{1-\gamma}}\\
&=\left[\varphi_{a,b,1-\gamma,\xi}^{\rm (cr)}[f;g]\right]^{\frac{1}{1-\gamma}}.
\end{split}
\end{equation*}
The inequality \eqref{ineq:down-Fisher} follows then by replacing the previous calculation, together with the identities \eqref{eq:interm5} and \eqref{eq:interm6}, in the inequality \eqref{ineq:bip-down}, completing the proof.
\end{proof}

\noindent \textbf{Remark. Equality in \eqref{ineq:down-Fisher}.} The equality condition in \eqref{ineq:down-Fisher} is inherited from the condition \eqref{eq:ineq-bip-down} applied to the transformed densities $\overline{\mathfrak{D}}_{a,b}[g]$ and $\mathfrak{D}_{a,b}[f]$, which gives
\begin{equation}\label{eq:interm13}
\overline{\mathfrak{D}}_{a,b}[g](s)\propto\mathfrak{D}_{a,b}[f](s)^{A}\left|\frac{d}{ds}\mathfrak{D}_{a,b}[f](s)\right|^B,
\end{equation}
with $A$ and $B$ defined in \eqref{eq:ineq-bip-down}. Taking into account the expression of the derivative of the biparametric down transformation (see \cite[Eq. (3.6)]{IP2025(c)}) which we recall below for the sake of completeness
$$
\frac{d}{ds}\mathfrak{D}_{a,b}[f](s)=bf(x)^{2a-2}|f'(x)|^{1-2b}\left(\frac{f(x)f''(x)}{(f'(x))^2}-\frac{a}{b}\right),
$$
we obtain after a substitution and easy algebraic manipulations the following condition for equality in \eqref{ineq:down-Fisher}:
$$
g(x)\propto f(x)^{1-a+Aa+2(a-1)B}|f'(x)|^{b-bA+(1-2b)B}\left|\frac{f(x)f''(x)}{(f'(x))^2}-\frac{a}{b}\right|^B.
$$

We close this section by noticing that the inequalities \eqref{ineq:bip-down} and \eqref{ineq:down-Fisher} can be seen as upper bounds for the R\'enyi divergence by quotients of a cross-Fisher and a Fisher information, respectively of a cross-down-Fisher and a down-Fisher measure, keeping the structure with levels of informational functionals and sharp inequalities introduced in \cite{IP2025(b)}.


\subsection{Up transformation}

Since in the previous section, by employing the biparametric down transformation, we established equalities bounding the R\'enyi divergence by functionals of Fisher information type, in this section we employ the up transformation, introduced in \cite{IP2025} (see also \cite{IP2025(c)}) in order to go one level above in the structure with levels of informational functionals and sharp inequalities introduced in \cite{IP2025(b)} and establish bounds of the R\'enyi divergence in terms of moment-like functionals. We first recall the definition of the up transformation (we restrict ourselves to the one-parameter one, avoiding the technical complications of the biparametric one):
\begin{definition}\label{def:up}
Let $\pdf: \supp\longrightarrow \mathbb R^+$ be a probability density function supported in the closure of $\Omega=(c,d)$. For $a\in\Rset\setminus\{2\}$, the up transformation is defined as
\begin{equation}\label{eq:up}
\mathfrak{U}_{a}[\pdf](y)=|(a-2)x(y)|^\frac{1}{2-a},\quad y'(x)=-|(a-2)x|^\frac{1}{a-2}f(x),
\end{equation}
while for $a=2$ the up transformation is defined as
\begin{equation}\label{eq:up2}
\mathfrak{U}_{2}[\pdf(x)](y)=e^{-x(y)},\quad y'(x)=-e^xf(x).
\end{equation}
\end{definition}
Let us mention here that $\mathfrak{U}_a$ is the inverse transformation to $\mathfrak{D}_a$, as proved in \cite[Proposition 2.3]{IP2025}. We introduce below the reciprocal transformation, according to the general framework in Section \ref{subsec:gf}.
\begin{definition}
Let $a\in\mathbb{R}$ and $f$, $g$ be two probability density functions satisfying \eqref{cond:support}. If $a\neq2$, we define
\begin{equation*}
	\overline{\mathfrak U}_{a}[g](y)=\frac{g(x)}{f(x)} |(2-a)x|^{\frac1{2-a}},\qquad y'(x)=-|(2-a)x|^{\frac1{a-2}}f(x),
\end{equation*}
while if $a=2$, the reciprocal transformation is given by
\begin{equation}\label{eq:rec-up2}
\overline{\mathfrak U}_{2}[g](y)=\frac{g(x)}{f(x)}e^{-x}, \quad y'(x)=-e^xf(x).
\end{equation}
\end{definition}
We next introduce, following the same pattern as in the previous sections, the cross-moment of $f$ relative to $g$.
\begin{definition}\label{def:cross-moment}
Let $p,\gamma\in\Rset$ and let $f$, $g$ be two probability density functions. The \emph{cross-deviation} of $f$ relative to $g$ is defined as
\begin{equation}\label{eq:cross-moment}
	\sigma_{p,\gamma}[f;g]:=\left(\int_\mathbb{R}[f(x)]^{2-\gamma}[g(x)]^{\gamma-1}\,|x|^p\,dx\right)^\frac1{p}.
\end{equation}
The \emph{exponential cross-deviation} of $f$ relative to $g$ is defined as
\begin{equation}\label{eq:cross-moment2}
\sigma^{(E)}_{\gamma}[f;g]:=\left(\int_{\Rset}[f(x)]^{2-\gamma}[g(x)]^{\gamma-1}e^{(1-\gamma)x}\,dx\right)^{\frac{1}{1-\gamma}}.
\end{equation}
\end{definition}
Let us note that, for $f=g$, $\sigma_{p,\gamma}[f;g]$ reduces to the standard $p$-deviation of the probability density function $f$. We are now in a position to establish a sharp inequality bounding the R\'enyi divergence by the cross-deviation and the deviation of a probability density function.
\begin{theorem}\label{th:up}
Let $\alpha$, $\beta$, $\gamma$ be three real numbers satisfying \eqref{eq:relation}, $f$ and $g$ be two probability density functions satisfying \eqref{cond:support} and $a\in\Rset$. If $a\neq2$, we have
\begin{equation}\label{ineq:up}
\left(\sigma_\frac{\alpha-1}{2-a}[f]\right)^\frac1{a-2}e^{D_\beta[f||g]}\leqslant\left(\sigma_{\frac{\gamma-1}{2-a},\gamma}[f;g]\right)^{\frac{1}{a-2}}.
	\end{equation}
If $a=2$, we have
\begin{equation}\label{ineq:up2}
\left\langle e^{(1-\alpha)x}\right\rangle_f^{\frac{1}{1-\alpha}}e^{D_{\beta}[f||g]}\leqslant\sigma_{\gamma}^{(E)}[f;g].
\end{equation}
The equality is achieved in the inequalities \eqref{ineq:up} and \eqref{ineq:up2} if and only if
\begin{equation}\label{eq:ineq-up}
g\propto\begin{cases}
          |x|^{\frac{1-\alpha}{(2-a)(\alpha-\beta)}}f(x), & \mbox{if } a\neq2, \\[1mm]
          e^{\frac{(\alpha-1)x}{\alpha-\beta}}f(x), & \mbox{if } a=2.
        \end{cases}
\end{equation}
\end{theorem}
\begin{proof}
Pick first $a\neq2$. It has been proved in \cite[Lemma 3.1]{IP2025} that
\begin{equation}\label{eq:interm7}
\exp\left(R_\alpha[\mathfrak U_{a}[f]]\right)=N_{\alpha}[\mathfrak U_{a}[f]]=\left(|2-a| \sigma_\frac{\alpha-1}{2-a}[f]\right)^\frac1{a-2}.
\end{equation}
Moreover, the equality
\begin{equation}\label{eq:interm8}
D_{\beta}[\mathfrak{U}_{a}[f]||\overline{\mathfrak{U}}_a[g]]=D_{\beta}[f||g]
\end{equation}
follows as a particular case of the general identity \eqref{eq:div}.
The inequality~\eqref{eq:ineqRRR} applied to $\mathfrak U_{a}[f]$ and $\overline{\mathfrak U_{a}}[g]$, together with \eqref{eq:interm7} and \eqref{eq:interm8}, give after taking exponentials
\begin{equation}\label{eq:interm9}
\left(|2-a|\sigma_\frac{\alpha-1}{2-a}[f]\right)^\frac1{a-2}e^{D_\beta[f||g]}\leqslant \exp\left(H_\gamma\left[\mathfrak U_{a}[f]\,||\,\overline{\mathfrak U_{a}}[g]\right]\right).
\end{equation}
We are only left to compute the right-hand side of the inequality \eqref{eq:interm9}. We have
\begin{eqnarray*}
\exp\left(H_\gamma\left[\mathfrak U_{a}[f]\,||\,\overline{\mathfrak U_{a}}[g]\right]\right)
&=&
\left(\int_\mathbb R
\mathfrak U_{a}[f](y)\left[\overline{\mathfrak U_{a}}[g]\right]^{\gamma-1}(y)dy
\right)^\frac1{1-\gamma}
\\
&=&
\left(\int_\mathbb R
f(x)\left[\frac{g(x)}{f(x)}|(2-a)x|^\frac{1}{2-a}\right]^{\gamma-1}dx
\right)^\frac1{1-\gamma}
\\
&=&
|2-a|^{\frac1{a-2}}\left(\int_\mathbb R
[f(x)]^{2-\gamma}[g(x)]^{\gamma-1}\, |x|^\frac{\gamma-1}{2-a}dx
\right)^\frac1{1-\gamma}\\
&=&|2-a|^{\frac1{a-2}}\sigma_{\frac{\gamma-1}{2-a},\gamma}^{\frac{1}{a-2}}[f;g].
\end{eqnarray*}
The inequality \eqref{ineq:up} follows by inserting the previous calculation into \eqref{eq:interm9}.

Let next $a=2$. Then, following the previous calculation but, in the final step of it, employing the definition \eqref{eq:rec-up2} for the reciprocal transformation, we find
\begin{equation*}
\begin{split}
\exp\left(H_\gamma\left[\mathfrak U_{2}[f]\,||\,\overline{\mathfrak U_{2}}[g]\right]\right)
&=\left(\int_\mathbb{R}f(x)\left[\frac{g(x)}{f(x)}e^{-x}\right]^{\gamma-1}dx\right)^{\frac1{1-\gamma}}\\
&=\left(\int_{\Rset}f(x)^{2-\gamma}g(x)^{\gamma-1}e^{(1-\gamma)x}\right)^{\frac{1}{1-\gamma}}=\sigma_{\gamma}^{(E)}[f;g].
\end{split}
\end{equation*}
Moreover, it has been shown in \cite[Eq. (3.6)]{IP2025} that
$$
\exp\left(R_\alpha[\mathfrak U_{2}[f]]\right)=N_{\alpha}[\mathfrak U_{2}[f]]=\left(\int_{\Rset}f(x)e^{(1-\alpha)x}\,dx\right)^{\frac{1}{1-\alpha}}=\left\langle e^{(1-\alpha)x}\right\rangle_f^{\frac{1}{1-\alpha}}.
$$
The inequality \eqref{ineq:up2} follows readily from the previous two calculations and \eqref{eq:interm8}. The equality is attained in \eqref{ineq:up} and \eqref{ineq:up2} if and only if $\overline{\mathfrak{U}}_{a}[g]$ and $\mathfrak{U}_a[f]$ satisfy the proportionality condition \eqref{eq:main}. Thus, for $a\neq2$ we find (dropping the constants according to the meaning of the notation $\propto$)
$$
\frac{g(x)}{f(x)}|x|^{\frac{1}{2-a}}\propto|x|^{\frac{1-\beta}{(2-a)(\alpha-\beta)}},
$$
leading to the first case in \eqref{eq:ineq-up}, while for $a=2$ we have
$$
\frac{g(x)}{f(x)}e^{-x}\propto e^{-\frac{x(1-\beta)}{\alpha-\beta}},
$$
leading to the second case in \eqref{eq:ineq-up} and completing the proof.
\end{proof}

\noindent\textbf{Remark.}  When $a\neq 2$, in the previous inequality the equality is not reached by a pair of exponential densities, but it does by a pair of power-law, or Pareto, densities. Furthermore, if we pick a Gaussian or a generalized normal density as $f$, then the equality is achieved when $g$ is a Weibull or a generalized Gamma density. In contrast to this case, when $a=2$, a pair of exponentials is a minimizing pair, but a pair of power law densities is no longer an equality pair.

\medskip

\noindent
The previous process can be iterated further, in order to derive inequalities for upper-moments of any order, following the process of iteration of applications of the up transformation presented in \cite[Section 3]{IP2025(b)}. Let us stress here that, while the down transformation cannot be in general iterated as many times as we wish (it requires more and more restrictive conditions of regularity on $f$ at every iteration step), the up transformation can be iterated $n$ times, for any natural number $n$. In the present work, we only perform one iteration, and leave the reader to iterate more and obtain further inequalities, along the same pattern, if needed. We also restrict ourselves to $a\neq2$ and $b\neq2$ in this step, in order to keep the presentation brief. Let us recall the definition and notation of the \emph{upper-moments} and \emph{upper-deviations} introduced in \cite[Definition 3.1]{IP2025(b)} for $a\in\Rset\setminus\{2\}$ and $p\in\Rset$:
\begin{equation}\label{eq:upper-mom}
M_{p,a}[\pdf]=\int_{\Rset}\left|\int_x^{d} |(a-2)t|^\frac1{a-2}f(t)dt \right|^p\pdf(x)dx, \quad m_{p,a}[f]=M_{p,a}[f]^\frac{a-2}p.
\end{equation}
We also introduce the \emph{cross-upper-moment} of two probability density functions $f$ and $g$ by the following expression:
\begin{equation}\label{def:cross-upper}
M_{p,\lambda,b}[f;g]:=\int_{\Rset}f(x)^{1-\lambda}g(x)^{\lambda}\left|\int_{x}^{d}|(b-2)t|^{\frac{1}{b-2}}f(t)\,dt\right|^p\,dx, \quad
m_{p,\lambda,b}[f;g]:=M_{p,\lambda,b}^{\frac{b-2}{p}},
\end{equation}
observing that $M_{p,\lambda,b}[f;f]=M_{p,b}[f]$. With this notation, we prove the following inequality.
\begin{theorem}\label{th:upper-mom}
Let $\alpha$, $\beta$, $\gamma$ be three real numbers satisfying \eqref{eq:relation} and $f$, $g$ be two probability density functions satisfying \eqref{cond:support}. Let $a$, $b\in\Rset\setminus\{2\}$. We then have:
\begin{equation}\label{ineq:upper}
m_{\frac{\alpha-1}{2-a},b}^{\frac{1}{(a-2)(b-2)}}[f]e^{D_{\beta}[f||g]}\leq m_{\frac{\gamma-1}{2-a},\gamma-1,b}^{\frac{1}{(a-2)(b-2)}}[f;g].
\end{equation}
The inequality \eqref{ineq:upper} is sharp and the equality is attained if and only if
\begin{equation}\label{eq:ineq-upper}
g(x)\propto f(x)\left|\int_x^{d}|t|^{\frac{1}{a-2}}f(t)\,dt\right|^{\frac{1-\alpha}{(2-a)(\alpha-\beta)}}.
\end{equation}
\end{theorem}
\begin{proof}
Pick $b\in\Rset\setminus\{2\}$. We apply \eqref{ineq:up} to the transformed densities $\mathfrak{U}_{b}[f]$, respectively $\overline{\mathfrak{U}}_b[g]$ (in place of $f$ and $g$) to obtain
\begin{equation}\label{eq:interm12}
\sigma_{\frac{\alpha-1}{2-a}}^{\frac{1}{a-2}}[\mathfrak{U}_b[f]]e^{D_{\beta}[\mathfrak{U}_b[f]||\overline{\mathfrak{U}}_b[g]]}\leq
\left(\sigma_{\frac{\gamma-1}{2-a},\gamma}[\mathfrak{U}_b[f];\overline{\mathfrak{U}}_b[g]]\right)^{\frac{1}{a-2}}.
\end{equation}
On the one hand, we have
\begin{equation*}
\sigma_{\frac{\alpha-1}{2-a}}^{\frac{1}{a-2}}[\mathfrak{U}_b[f]]=\mu_{\frac{\alpha-1}{2-a}}^{\frac{1}{1-\alpha}}[\mathfrak{U}_b[f]]
=M_{\frac{\alpha-1}{2-a},b}^{\frac{1}{1-\alpha}}[f]=m_{\frac{\alpha-1}{2-a},b}^{\frac{1}{(a-2)(b-2)}}[f].
\end{equation*}
On the other hand, we can compute the right-hand side of \eqref{eq:interm12} as follows from Eq~\eqref{eq:cross-moment}:
\begin{equation*}
\begin{split}
\left(\sigma_{\frac{\gamma-1}{2-a},\gamma}[\mathfrak{U}_b[f];\overline{\mathfrak{U}}_b[g]]\right)^{\frac1{a-2}}
&=\left[\int_{\Rset}\mathfrak{U}_b[f](y)\left(\frac{\overline{\mathfrak{U}}_b[g]}{\mathfrak{U}_b[f]}\right)^{\gamma-1}(y)
|y|^{\frac{\gamma-1}{2-a}}\,dy\right]^{\frac{1}{1-\gamma}}\\
&=\left[\int_{\Rset}f(x)\left(\frac{g(x)}{f(x)}\right)^{\gamma-1}
\left|\int_{x}^{d}|(2-b)t|^{\frac{1}{b-2}}f(t)\,dt\right|^{\frac{\gamma-1}{2-a}}\,dx\right]^{\frac{1}{1-\gamma}}\\
&=m_{\frac{\gamma-1}{2-a},\gamma-1,b}^{\frac{1}{(a-2)(b-2)}}[f;g].
\end{split}
\end{equation*}
The proof of the inequality \eqref{ineq:upper} is completed by the already standard fact that
$$
D_{\beta}[\mathfrak{U}_b[f]||\overline{\mathfrak{U}}_b[g]]=D_{\beta}[f||g],
$$
following from \eqref{eq:div}. Finally, the equality is achieved in \eqref{ineq:upper} if and only if $\overline{\mathfrak{U}}_a[g]$ and $\mathfrak{U}_a[f]$ satisfy the first case of the equality condition \eqref{eq:ineq-up}, that is,
$$
\overline{\mathfrak{U}}_a[g](y)\propto\mathfrak{U}_a[f](y)|y|^{\frac{1-\alpha}{(2-a)(\alpha-\beta)}},
$$
where $y$ is the new independent variable introduced in the up transformation. The latter relation implies
$$
\frac{g(x)}{f(x)}\propto|y(x)|^{\frac{1-\alpha}{(2-a)(\alpha-\beta)}},
$$
which immediately leads to \eqref{eq:ineq-upper} by replacing $y(x)$ by its integral formula stemming from \eqref{eq:up}, ending the proof.
\end{proof}
It is rather obvious how one can iterate further the definition of cross-upper-moments of higher order following the same pattern as in \eqref{def:cross-upper} and how the inequality \eqref{ineq:upper} writes for higher order upper-moments. We leave this easy extension to the reader.


\subsection*{Acknowledgements}

R. G. I. is partially supported by the project PID2024-160967NB-I00 (AEI) funded by the Ministry of Science, Innovation and Universities of Spain and FEDER/EU. D. P.-C. is partially supported by the project PID2023-153035NB-100 (AEI) funded by the Ministry of Science, Innovation and Universities of Spain and “ERDF/EU A way of making Europe”.	

\bigskip

\noindent \textbf{Data availability} Our manuscript has no associated data.

\bigskip

\noindent \textbf{Competing interest} The authors declare that there is no competing interest.

		\bibliographystyle{unsrt}
		\bibliography{refs}

\begin{thebibliography}{10}

\bibitem{Abe2003}
S.~Abe.
\newblock Geometry of escort distributions.
\newblock {\em Physical Review E}, 68(3):031101, 2003.

\bibitem{Bercher2012(divergences)}
J.-F. Bercher.
\newblock A simple probabilistic construction yielding generalized entropies
  and divergences, escort distributions and q-gaussians.
\newblock {\em Physica A: Statistical Mechanics and its Applications},
  391(19):4460--4469, 2012.

\bibitem{Tsallis2009(book)}
C.~Tsallis.
\newblock {\em Introduction to {N}onextensive {S}tatistical {M}echanics:
  {A}pproaching a {C}omplex {W}orld}, volume~1.
\newblock Springer, 2009.

\bibitem{Fisher1925}
R.~A. Fisher.
\newblock Theory of statistical estimation.
\newblock In {\em Mathematical {P}roceedings of the Cambridge {P}hilosophical
  {S}ociety}, volume~22, pages 700--725. Cambridge University Press, 1925.

\bibitem{Hammad1978}
P.~Hammad.
\newblock Mesure d’ordre $\alpha$ de l’information au sens de {F}isher.
\newblock {\em Revue de statistique appliqu{\'e}e}, 26(1):73--84, 1978.

\bibitem{Antolin2009}
J.~Antol{\'\i}n, J.~C. Angulo, and S.~L{\'o}pez-Rosa.
\newblock Fisher and {J}ensen--{S}hannon divergences: Quantitative comparisons
  among distributions. application to position and momentum atomic densities.
\newblock {\em The Journal of {C}hemical Physics}, 130(7), 2009.

\bibitem{Martin2013}
A.~L. Mart{\'\i}n, J.~C. Angulo, and J.~Antol{\'\i}n.
\newblock Fisher-like atomic divergences: Mathematical grounds and physical
  applications.
\newblock {\em Physica A: Statistical Mechanics and its Applications},
  392(21):5552--5563, 2013.

\bibitem{Antolin2014}
J.~Antol{\'\i}n, J.~C. Angulo, S.~Mulas, and S.~L{\'o}pez-Rosa.
\newblock Relativistic global and local divergences in hydrogenic systems: A
  study in position and momentum spaces.
\newblock {\em Physical Review A}, 90(4):042511, 2014.

\bibitem{Yamano2021}
T.~Yamano.
\newblock Skewed {J}ensen--{F}isher divergence and its bounds.
\newblock {\em Foundations}, 1(2):256--264, 2021.

\bibitem{Lutwak2005}
E.~Lutwak, D.~Yang, and G.~Zhang.
\newblock Cram\'er--{R}ao and moment-entropy inequalities for \text{R}\'enyi
  entropy and generalized {F}isher information.
\newblock {\em IEEE Transactions on Information Theory}, 51(2):473--478, 2005.

\bibitem{Bercher2012}
J.-F. Bercher.
\newblock On generalized {C}ram{\'e}r--{R}ao inequalities, generalized {F}isher
  information and characterizations of generalized {$q-$}{G}aussian
  distributions.
\newblock {\em Journal of Physics A: Mathematical and Theoretical},
  45(25):255303, 2012.

\bibitem{Toranzo2018}
E.~V. Toranzo, S.~Zozor, and J.-M. Brossier.
\newblock Generalization of the de {B}ruijn identity to general
  $\phi$-entropies and $\phi$-{F}isher informations.
\newblock {\em IEEE Transactions on Information Theory}, 64(10):6743--6758,
  2018.

\bibitem{Tempesta2011}
P.~Tempesta.
\newblock Group entropies, correlation laws, and zeta functions.
\newblock {\em Physical Review E—Statistical, Nonlinear, and Soft Matter
  Physics}, 84(2):021121, 2011.

\bibitem{Tempesta2019}
M.~{\'A}. Rodr{\'\i}, {\'A}.~Romaniega, and P.~Tempesta.
\newblock A new class of entropic information measures, formal group theory and
  information geometry.
\newblock {\em Proceedings of the Royal Society A}, 475(2222):20180633, 2019.

\bibitem{IPT2025}
R.~G. Iagar, D.~Puertas-Centeno, and E.~V. Toranzo.
\newblock Sharp informational inequalities involving {K}ullback-{L}eibler and
  {R}ényi divergences and a family of scaling-invariant relative {F}isher
  measures.
\newblock {\em arXiv preprint arXiv:2507.17408}, 2025.

\bibitem{Lin2002}
J.~Lin.
\newblock Divergence measures based on the {S}hannon entropy.
\newblock {\em IEEE Transactions on Information Theory}, 37(1):145--151, 2002.

\bibitem{Sanchez-Moreno2012}
P.~S{\'a}nchez-Moreno, A.~Zarzo, and J.~S. Dehesa.
\newblock Jensen divergence based on {F}isher's information.
\newblock {\em Journal of Physics A: Mathematical and Theoretical}, 45, 2012.

\bibitem{Stummer-Vajda2012}
W.~Stummer and I.~Vajda.
\newblock On {B}regman distances and divergences of probability measures.
\newblock {\em IEEE Transactions on Information Theory}, 58(3):1277--1288,
  2012.

\bibitem{Zozor2017}
S.~Zozor, D.~Puertas-Centeno, and J.~S. Dehesa.
\newblock On generalized {S}tam inequalities and {F}isher--\text{R}{\'e}nyi
  complexity measures.
\newblock {\em Entropy}, 19(9):493, 2017.

\bibitem{Puertas2025}
D.~Puertas-Centeno and S.~Zozor.
\newblock Some informational inequalities involving generalized trigonometric
  functions and a new class of generalized moments.
\newblock {\em Journal of Physics A: Mathematical and Theoretical},
  58(16):165002, 2025.

\bibitem{IP2025}
R.~G. Iagar and D.~Puertas-Centeno.
\newblock A new pair of transformations and applications to generalized
  informational inequalities and {H}ausdorff moment problem.
\newblock {\em Communications in Nonlinear Science and Numerical Simulation},
  151:109091, 2025.

\bibitem{IP2025(b)}
R.~G. Iagar and D.~Puertas-Centeno.
\newblock Through and beyond moments, entropies and {F}isher information
  measures: new informational functionals and inequalities.
\newblock {\em Physica D: Nonlinear Phenomena}, 483:134928, 2025.

\bibitem{IP2025(c)}
R.~G. Iagar and D.~Puertas-Centeno.
\newblock Generalized informational functionals and new monotone measures of
  statistical complexity.
\newblock {\em arXiv preprint arXiv:2511.02502}, 2025.

\bibitem{GPPT25(b)}
P.~R. Gordoa, A.~Pickering, D.~Puertas-Centeno, and E.~V. Toranzo.
\newblock Sundman-like transformations and the {NRT} nonlinear
  {S}chr{\"o}dinger equation.
\newblock {\em arXiv preprint arXiv:2511.11765}, 2025.

\bibitem{Beck2004}
C.~Beck.
\newblock Superstatistics, escort distributions, and applications.
\newblock {\em Physica A: Statistical Mechanics and its Applications},
  342(1-2):139--144, 2004.

\bibitem{Bercher2011}
J.-F. Bercher.
\newblock Escort entropies and divergences and related canonical distribution.
\newblock {\em Physics Letters A}, 375(33):2969--2973, 2011.

\bibitem{Puertas2019}
D.~Puertas-Centeno.
\newblock Differential-escort transformations and the monotonicity of the
  {LMC}-\text{R}{\'e}nyi complexity measure.
\newblock {\em Physica A: Statistical Mechanics and its Applications},
  518:177--189, 2019.

\bibitem{Bercher2012a}
J.-F. Bercher.
\newblock On a {$(\beta,q)-$}generalized {F}isher information and inequalities
  involving {$q-$}{G}aussian distributions.
\newblock {\em Journal of Mathematical Physics}, 53(6), 2012.

\end{thebibliography}

\end{document}